\newif\ifarxiv
\arxivfalse
\documentclass[fontsize=11pt,a4paper,DIV=12]{scrartcl}

\usepackage[utf8]{inputenc}

\usepackage{graphicx}
\graphicspath{{figs/}}

\usepackage{xcolor}
\usepackage{amssymb}
\usepackage{amsmath}
\usepackage{amsthm}
\usepackage{comment}
\usepackage{thm-restate}

\usepackage[font=small,labelfont=bf]{caption}
\usepackage[font=small,labelfont=normalfont,labelformat=simple]{subcaption}

\usepackage{url}
\makeatletter
\AtBeginDocument{
  \def\doi#1{\url{https://doi.org/#1}}}
\makeatother

\newtheorem{theorem}{Theorem}
\newtheorem{proposition}[theorem]{Proposition}

\newtheorem{problem}{Problem}
\newtheorem{conjecture}[theorem]{Conjecture}

\newcounter{claimCount}
\newenvironment{claim}{\medskip
\noindent\refstepcounter{claimCount}\textbf{Claim~\arabic{claimCount}.}}{\medskip}
\AtBeginEnvironment{proof}{\setcounter{claimCount}{0}}

\newtheorem*{observation}{Observation}

\def\AA{{\cal A}}

\def\AAsixA{{\cal N}_6^\Delta}
\def\DD{{\cal D}}
\def\KK{{\cal K}}
\def\TG{{T}}

\def\qedclaim{\hfill$\triangle$}

\usepackage[hidelinks]{hyperref}
\newcommand{\orcidID}[1]{\href{https://orcid.org/#1}{\includegraphics[scale=.03]{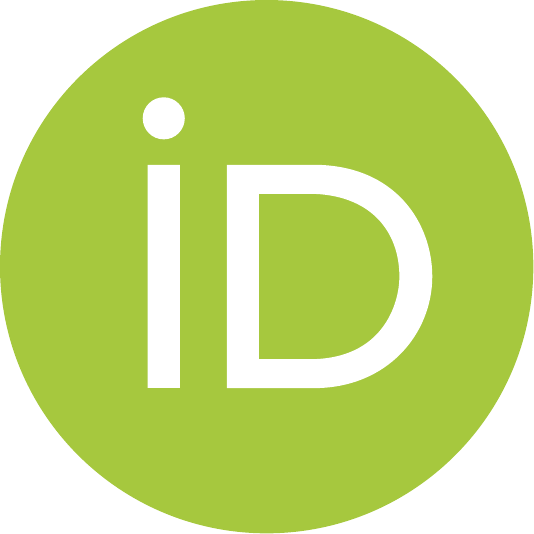}}}

\title{
    Arrangements of Pseudocircles: \\
    On Digons and Triangles\footnote{
        A part of this work was initiated at a workshop of the collaborative DACH project \emph{Arrangements and  Drawings} in Gathertown. We thank the organizers and all the participants for the  inspiring atmosphere. S.~Roch was funded by the DFG-Research-Training-Group 'Facets of Complexity' (DFG-GRK~2434). M.~Scheucher was supported by the DFG Grant SCHE~2214/1-1.
    }
}
    
\author{
    Stefan Felsner\orcidID{0000-0002-6150-1998}
    \and
    Sandro Roch\orcidID{0000-0002-9353-9413}
    \and
    Manfred Scheucher\orcidID{0000-0002-1657-9796}
}
\date{}
\begin{document}

\maketitle
\vspace{-1.5cm}
\begin{center}
        {\small
                Institut f\"ur Mathematik, \\
                Technische Universit\"at Berlin, Germany,\\
                \texttt{$\{$felsner,roch,scheucher$\}$@math.tu-berlin.de}
                \\
        }
\end{center}
\vspace{0.25cm}

    \begin{abstract}				
        In this article, we study the cell-structure of simple arrangements of pairwise intersecting pseudocircles. The focus will be on two problems of Gr\"unbaum (1972).
    	
        First, we discuss the maximum number of digons or touching points. Gr\"unbaum conjectured that there are at most~$2n-2$ digon cells or equivalently at most~$2n-2$ touchings.  Agarwal~et~al.~(2004) verified the conjecture for cylindrical arrangements. We show that the conjecture holds for any arrangement which contains three pairwise touching pseudocircles. The proof makes use of the result for cylindrical arrangements. Moreover, we construct non-cylindrical arrangements which attain the maximum of~$2n-2$ touchings and have no triple of pairwise touching pseudocircles.
    
        Second, we discuss the minimum number of triangular cells (triangles) in arrangements without digons and touchings. Gr\"unbaum conjectured that such arrangements have $2n-4$ triangles. Snoeyink and Hershberger (1991) established a lower bound of $\lceil \frac{4}{3}n \rceil$. Felsner and Scheucher (2017) disproved the conjecture and constructed a family of arrangements with only $\lceil \frac{16}{11}n \rceil$ triangles. We provide a construction which shows that $\lceil \frac{4}{3}n \rceil$ is the correct value.
        \medskip

        \noindent{\small\textbf{Keywords.}
        \def\and{--\ }
            arrangement of pseudocircles \and touching \and empty lense \and cylindrical arrangement \and arrangement of pseudoparabolas \and Grünbaum's conjecture
        }
    \end{abstract}
    
\section{Introduction}
    \label{sec:introduction}
	
	An \emph{arrangement $\AA$ of pairwise intersecting pseudocircles} is a collection of $n(\AA)$ simple closed curves on the sphere or plane such that any two of the curves either touch in a single point or intersect in exactly two points where they cross. 
	Throughout this article, we consider all arrangements to be \emph{simple}, that is, no three pseudocircles meet in a common point. An arrangement $\AA$ partitions the plane into cells. A cell with exactly $k$ crossings on its boundary is a \emph{$k$-cell},
	\mbox{$2$-cells} are also called \emph{digons} and $3$-cells are \emph{triangles}. The number of $k$-cells of an arrangement $\AA$ is denoted as~$p_k(\AA)$. 
	
	The study of cells in arrangements started almost a century ago when Levi~\cite{Levi1926} showed that, in an arrangement of at least three pseudolines in the projective plane, every pseudoline is incident to at least three triangles. In the 1970's, Grünbaum~\cite{Gruenbaum1972} studied arrangements of pseudolines and initiated the study of arrangements of pseudocircles.

	\subsection{Digons and touchings}
	
	Concerning digons in arrangements of pairwise intersecting pseudocircles, 
	Grün\-baum \cite{Gruenbaum1972} 
	presented a construction with $2n-2$ digons 
	(depicted in Figure~\ref{fig:gruenbaum_2n_2_digon_construction})
	and	conjectured that these arrangements have the maximum number of digons\footnote{Originally the conjecture was stated as to include non-simple arrangements which are \emph{non-trivial}, i.e., non-simple arrangements with at least 3 crossing points.}.

	\begin{conjecture}[Grünbaum's digon conjecture \hbox{\cite[Conjecture~3.6]{Gruenbaum1972}}]
		\label{conjecture:digons}
		Every simple arrangement~$\AA$ of $n$ pairwise intersecting pseudocircles has at most $2n-2$ digons, i.e., $p_2 \le 2n-2$.
	\end{conjecture}

    \begin{figure}[htb]
      \centering
      \includegraphics{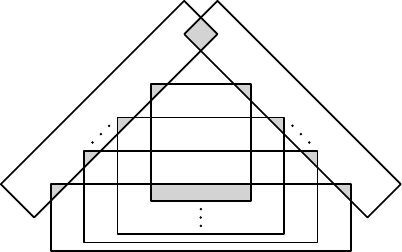}
      \caption{An arrangement of $n\ge 4$ pairwise intersecting pseudocircles with exactly $2n-2$ digons. Digons are highlighted gray (Example copied from Gr\"unbaum~\cite[Figure~3.28]{Gruenbaum1972}).}
  \label{fig:gruenbaum_2n_2_digon_construction}
    \end{figure}	

    It was shown by Agarwal et al.\ \cite{AgarwalNPPSS2004} that Conjecture \ref{conjecture:digons} holds for simple cylindrical arrangements.

	An intersecting arrangement of pseudocircles is \emph{cylindrical} if there is a pair of cells which are separated by each pseudocircle of the arrangement. 
	An \emph{intersecting arrangement of pseudoparabolas} is a collection of infinite $x$-monotone curves, called \emph{pseudoparabolas}, where each pair of them either have a single touching or intersect in exactly two points where they cross. Every cylindrical arrangement of pseudocircles can be represented as an arrangement of pseudoparabolas and vice versa.
    	From an arrangement of pseudoparabolas one can directly obtain a drawing of
	an arrangement of pseudocircles on the lateral surface of a cylinder so that the pseudocircles wrap around the cylinder. 
	The two separating cells correspond to the top and the bottom of the cylinder.
 The top and bottom cell may have degree two, they account for the difference of two between the conjecture and the result for pseudo-parabolas.
 
   More precisely  Agarwal et al.~\cite{AgarwalNPPSS2004} showed that the number of touchings in an intersecting arrangement of~$n$~pseudo-parabolas is at most $2n-4$ [Theorem~2.4] and in an intersecting cylindrical arrangement it is at most $2n-3$ [Corollary~2.12]. They extended this by showing that intersecting arrangements of pseudocircles the number of digons is at most linear in~$n$~[Theorem~2.13].
	The proof of the last result is based on the fact that every arrangement of intersecting pseudocircles can be stabbed by constantly many points. That is, there exists an absolute constant~$k$, called the \emph{stabbing number}\footnote{In the literature, the stabbing number is also referred to as \emph{piercing number} or \emph{transversal number}.}, such that for every arrangement of $n$ pseudocircles in the plane there exists a set of $k$ points with the property that each pseudocircle contains at least one of the points in its interior region~\cite[Corollary~2.8]{AgarwalNPPSS2004}. Therefore, the arrangement can be decomposed into constantly many cylindrical subarrangements.
	The linear upper bound then follows from the fact that each pair of subarrangements contributes at most linearly many digons. 
	Grünbaum's digon conjecture is known to hold for arrangements with up to $7$ pseudocircles; see~\cite{FelsnerScheucher2020}. 

 The conjecture has also been studied for proper circles. Alon et al.~\cite{AlonLPS01}
 proved that every arrangement of $n$ pairwise intersecting circles contains 
 at most $20n-2$ digons. For arrangements of pairwise intersecting unit circles, Pinchasi~\cite{Pinchasi2002} proved that there are at most~$n + 3$ digons.
 Very recently Ackerman et al.~\cite{AckermanDKPR24} verified Grünbaums digon conjecture for circles.
 
	\medskip
	
In this paper we show that Grünbaum's digon conjecture (Conjecture~\ref{conjecture:digons}) holds for arrangements which contain three pseudocircles that pairwise form a digon. Before we state the result as a theorem, let us introduce some notation. For an arrangement $\AA$ of pseudocircles and any selection of its digons, we can perform a perturbation so that the selected digons become touching points. Figure~\ref{fig:digon_contract} gives an illustration. It is therefore sufficient to find an upper bound on the number of touchings to prove Grünbaum's digon conjecture. We define the \emph{touching graph}~$\TG(\AA)$ to have the pseudocircles of $\AA$ as vertices, and two vertices form an edge if the two corresponding pseudocircles touch.
	
\begin{figure}[htb]
    \centering
    \includegraphics[page=2]{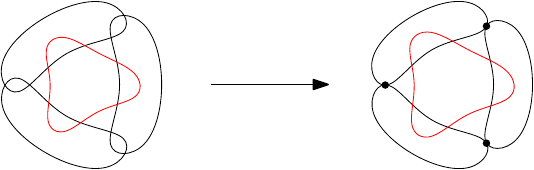}
    \caption{Contracting some digons to touchings.}
    \label{fig:digon_contract}
\end{figure}

\begin{restatable}{theorem}{thmDigons}
    \label{thm:digons}
    Let $\AA$ be a simple arrangement of $n$ pairwise intersecting pseudocircles. If the touching graph $\TG(\AA)$ contains a triangle, then there exist at most~$2n-2$ touchings, i.e.,~\mbox{$p_2(\mathcal{A}) \le 2n-2$}.
\end{restatable}
	
Theorem \ref{thm:digons} in particular shows that Grünbaum's construction with $2n-2$ touchings is maximal for arrangements with triangles in the touching graph. However, the maximum number of touchings in general arrangements remains unknown. In Section~\ref{sec:prop:trianglefree_tight} we construct a family of arrangements of $n$ pairwise intersecting pseudocircles which have exactly $2n-2$ touchings and a triangle free touching graph. This family witnesses that the conjectured upper bound (Conjecture \ref{conjecture:digons}) can also be achieved in the cases not covered by Theorem~\ref{thm:digons}. 
	
\begin{restatable}{proposition}{propTriangleFreeTight}
    \label{prop:trianglefree_tight}
    For $n\in\{11, 14, 15\}$ and $n \geq 17$ there exists a simple arrangement $\AA_n$ of~$n$ pairwise intersecting pseudocircles with no triangle in the touching graph $\TG(\AA_n)$ and with exactly $p_2(\AA_n)=2n-2$ touchings.
\end{restatable}

\subsection{Triangles in digon-free arrangements}
	
    In this context we assume that all arrangements are digon- and touching-free. It was shown by Levi \cite{Levi1926} that every arrangement of $n$ pseudolines in the projective plane contains at least~$n$ triangles. Since arrangements of pseudolines are in correspondence with arrangements of \emph{great-pseudocircles} (see e.g.\ \cite[Section~4]{FelsnerScheucher2019}), it follows directly that an arrangement of~$n$ great-pseudocircles contains at least~$2n$ triangles, i.e.,~$p_3 \ge 2n$.
	
    Grünbaum conjectured that every digon-free intersecting arrangement on~$n$ pseudocircles contains at least $2n-4$ triangles \cite[Conjecture~3.7]{Gruenbaum1972}. Snoeyink and Hershberger \cite{SnoeyinkHershberger1991} proved a sweeping lemma for arrangements of pseudocircles. Using this powerful tool, they concluded that in every digon-free intersecting arrangement every pseudocircle has two triangles on each of its two sides (interior and exterior). This immediately implies the lower bound~$p_3(\AA) \ge 4n/3$; see Section~4.2 in \cite{SnoeyinkHershberger1991}.
	
    An infinite family of  digon-free arrangements of intersecting pseudocircles with~\mbox{$p_3 < \frac{16}{11}n$} was constructed in~\cite{FelsnerScheucher2020}. This family shows that Grünbaum's conjecture is wrong. With computer assistence~\cite{FelsnerScheucher2020}, it was also verified that the lower bound $p_3 \ge 4n/3$ is tight for~\mbox{$6 \le n \le 14$}. Here we show that the bound is tight for all $n\ge6$:
	
    \begin{restatable}{theorem}{thmTriangles}\label{thm:triangles}
        For every $n \ge 6$, there exists a simple digon-free arrangement $\AA_n$ of $n$ pairwise intersecting pseudocircles with $p_3(\AA_n)=\lceil \frac{4}{3}n \rceil$ triangles. Moreover, these arrangements are cylindrical.
    \end{restatable}

    The construction that we use for proving Theorem~\ref{thm:triangles} is based on replacing pseudocircles by bundles of pseudocircles. Starting from any digon-free arrangement, we can extend it to a counterexample of Grünbaum's triangle conjecture:

    \begin{restatable}{theorem}{thmExtensionToCounterexample}\label{thm:extension_to_counterexample}
    For $\varepsilon > 0$ fixed, every digon-free arrangement $\AA$ of $n$ pairwise intersecting pseudocircles is contained as a subarrangement in a digon-free arrangement $\hat{\AA}$ of $\hat{n}$ pairwise intersecting pseudocircles with $p_3(\hat{\AA}) < (\frac{4}{3} + \varepsilon)\hat{n}$.
    \end{restatable}

    For $n=6$ there is a unique intersecting digon-free arrangement which minimizes the number of triangles. This arrangement $\AAsixA$ (see Figure~\ref{fig:AAsixA})
    has been shown to be non-circularizable 
    in~\cite{FelsnerScheucher2019}, i.e., $\AAsixA$ cannot be represented as an arrangement of proper circles. 
    All counterexamples to Grünbaums conjecture presented in~\cite{FelsnerScheucher2020} as well as
    the arrangements constructed in the proof of Theorem~\ref{thm:triangles} contain  $\AAsixA$ as a subarrangement. Hence,
    all these arrangements 
    are non-circularizable. 
    The following weakening of the original 
    Grünbaum conjecture has been stated as 
 Conjecture~2.2 in~\cite{FelsnerScheucher2020}.
 
    \begin{conjecture}[Weak Grünbaum triangle conjecture]\label{conj:weakGB}
        Every simple digon-free arrangement~$\AA$ of $n$ pairwise intersecting circles has at least~$2n-4$ triangles.
    \end{conjecture}
	
To prove the conjecture it would be enough to verify that every simple intersecting digon-free arrangement of $n$ pseudocircles with less than $2n-4$ triangles contains $\AAsixA$ as a subarrangement. This, however, is wrong.
There are counterexamples to Grünbaum's conjecture 
without~$\AAsixA$ as a subarrangement. In Subsection~\ref{ssec:no-N6}, we prove the following proposition and discuss additional constructions.

\begin{restatable}{proposition}{propNoSubNSixFamily}\label{prop:nosubN6family}
    There is an infinite family of simple intersecting digon-free arrangements of~$n$ pseudocircles with~$\lceil\frac{5}{3}n\rceil + 2$ triangles which have no subarrangement isomorphic to $\AAsixA$.
\end{restatable}

\subsection{Related Work and Discussion}

In the proof of Theorem~\ref{thm:digons},
	we make use of a triangle ($K_3$) in the touching graph
	to bound the number of digons in the arrangement.
	It would be interesting to know whether other subgraphs like $C_4$ or $K_{3,3}$ can also be used to bound the number of digons. 
	
	The focus of this article is on arrangements of pairwise intersecting pseudocircles.
	For the setting of arrangements, where pseudocircles do not necessarily pairwise intersect,
	a classical construction of Erd\H{o}s~\cite{Erdos1946}
	gives arrangements of $n$ unit circles
	with $\Omega(n^{1+c/\log\log n})$ touchings. 
	An upper bound of $O(n^{3/2+\epsilon})$ on the number of digons in \emph{circle} arrangements was shown
	by Aronov and Sharir~\cite{AronovSharir2002}.
	The precise asymptotics, however, remain unknown.
	Moreover, we are not aware of an upper bound for \emph{pseudocircles}.
	\begin{problem}
    	Determine the maximum number of touchings among all simple non-inter\-secting arrangements 
            of~$n$ circles and pseudocircles, respectively.
	\end{problem}
	
	\medskip
	
	Concerning intersecting arrangements with digons, the number of triangles behaves  
	different than in digon-free arrangements. While the lower 
        bound is $p_3 \ge 2n/3$, we know that in the range of
        $3\le n \le 7$ the correct bound is $p_3 \ge n-1$, this was obtained using a computer-assisted exhaustive enumeration~\cite{FelsnerScheucher2020}.
	This motivates the conjecture:
	
	\begin{conjecture}[{}{\cite[Conjecture~2.10]{FelsnerScheucher2020}}]
		\label{conj:n_minus_1_triangles}
		Every simple arrangement of $n \ge 3$ pairwise intersecting pseudocircles has at least $n-1$ triangles, i.e., $p_3 \ge n-1$.
	\end{conjecture}

   Dropping the condition that the arrangements are intersecting seems non-interesting at first:
   arrangements of pairwise non-intersecting circles have no triangles. If we add the condition that the 
   intersection graph of the circles is connected we still have arrangements with a bipartite intersection graph where all faces are of even length, hence, there are no triangles. In the case where the intersection graph is connected and digons are forbidden triangles are unavoidable, in fact for all $n\geq 3$ there are arrangements in this class with only~8 triangles and this is the minimum. 
   
	Concerning the maximum number of triangles in intersecting arrangements, Felsner and Scheucher~\cite{FelsnerScheucher2020} 
	have shown an upper bound $p_3 \le \frac{4}{3} \binom{n}{2} + O(n)$
	which is optimal up to a linear error term.
	In fact,
	while $\frac{4}{3} \binom{n}{2} $ is an upper bound for arrangements of great-pseudocircles,
	we found an intersecting arrangement ($n=7$)
	with no digons, no touchings,
	and $29 = \frac{4}{3} \binom{7}{2}+1$ triangles. We are not aware of an infinite family of such arrangements.
\begin{problem}
  Determine the maximum number of 
  triangles of simple arrangements of~$n$ pairwise intersecting pseudocircles.
\end{problem}

	\section{Proof of Theorem~\ref{thm:digons}}
	\label{sec:thm:digons}

\thmDigons*

\begin{proof}
	
	Since the touching graph $\TG(\AA)$ contains a triangle, there are three pseudocircles in $\AA$ that pairwise touch. 
	Let $\KK$ be the subarrangement induced by these three pseudocircles 
	and let~$\triangle$ and~$\triangle'$ denote the two open triangle cells in~$\KK$. We label the three touching points, which are also the vertices of $\triangle$ and $\triangle'$, as $a, b, c$. Furthermore, we label the three boundary arcs of $\triangle$ (resp.~$\triangle'$) as $\alpha, \beta, \gamma$ (resp.~$\alpha', \beta', \gamma'$), as shown in Figure~\ref{fig:K3_labels_a}.
		
	\begin{figure}[htb]
	\centering
	
	\begin{subfigure}[b]{.32\textwidth}
		\centering
		\includegraphics[page=3]{figs/K3_two_types}
		\caption{}
		\label{fig:K3_labels_a}
	\end{subfigure}
	\hfill
	\begin{subfigure}[b]{.32\textwidth}
		\centering
		\includegraphics[page=1]{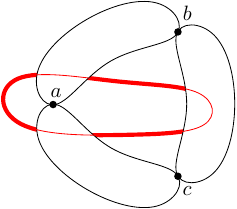}
		\caption{}
		\label{fig:K3_labels_b}
	\end{subfigure}
	\hfill
	\begin{subfigure}[b]{.32\textwidth}
		\centering
		\includegraphics[page=2]{figs/K3_two_types}
		\caption{}
		\label{fig:K3_labels_c}
	\end{subfigure}
	
	\caption{
		\subref{fig:K3_labels_a}~An illustration of the  subarrangement~$\KK$.
		\subref{fig:K3_labels_b}~and~\subref{fig:K3_labels_c}, respectively, illustrate an additional pseudocircle $C$ (red),
		the pc-arcs inside $\triangle\cup\triangle'$ are highlighted.
	}
	\label{fig:K3_labels}
	\end{figure}

	Assume that all digons in $\AA$ are contracted to touchings. In the following, the arrangement in Figure~\ref{fig:run_example:original_arrangement} will serve as a running example for~$\AA$. The intersection of a pseudocircle~$C\in\AA \setminus \KK$ with $\triangle \cup \triangle'$ results in three connected segments, which we denote as the three \emph{pc-arcs} of~$C$; see Figures~\ref{fig:K3_labels_b} and~\ref{fig:K3_labels_c}. 
	Note that 
	two of the pc-arcs induced by~$C$ may share an endpoint 
	if~$C$ forms a touching with one of the pseudocircles from~$\KK$; in the example arrangement in Figure~\ref{fig:run_example:original_arrangement}, this occurs~$5$~times on the boundary of~$\triangle$ and~$4$~times on the boundary of~$\triangle'$. 
	
	Each pc-arc 
	in $\triangle$ connects two of $\alpha, \beta$ or $\gamma$ while a pc-arc in  $\triangle'$ connects two of $\alpha', \beta'$ and~$\gamma'$. Depending on the boundary arcs on which they start and end, they belong to one of the types~$\alpha\beta$, $\beta\gamma$, $\alpha\gamma$, $\alpha'\beta'$, $\beta'\gamma'$ or $\alpha'\gamma'$. Figure~\ref{fig:run_example:original_triangles} shows the regions of~$\triangle$ and~$\triangle'$ together with the pseudocircles passing through them; Figure~\ref{fig:run_example:arc_type_triangles} shows the same regions but the arcs are colored according to their type in blue~($\alpha\beta$,~$\alpha'\beta'$), red~($\beta\gamma$,~$\beta'\gamma'$) or blue~($\alpha\gamma$,~$\alpha'\gamma'$).

	\begin{claim}
	\label{claim:same_arc_type}
	If two pc-arcs inside $\triangle$ (resp.\ $\triangle'$) have a touching or cross twice, then they are of the same type. 
	\end{claim}
	
    \noindent\textit{Proof of Claim~\ref{claim:same_arc_type}.} We prove the claim for $\triangle$; the argument for $\triangle'$ is the same.
	Suppose towards a contradiction that two distinct pseudocircles $C, C'$ from $\AA\setminus\KK$ contain pc-arcs~$A\subset C\cap\triangle$ and~$A'\subset C'\cap\triangle$ of different types that have a touching or cross twice.  
	For simplicity, consider only the arrangement induced by the five pseudocircles $\mathcal{K}\cup\{C, C'\}$. By symmetry we may assume that $A$ is of type $\alpha\gamma$ and $A'$ is of type $\alpha\beta$. We may further assume that $A$ and $A'$ have a touching, since otherwise, if they cross twice, they form a digon and we can contract it. This allows us to distinguish four cases which are depicted in Figure~\ref{fig:K3_arc_type_proof} (up to further possible contractions of digons formed between $C$ and the pseudocircles~of~$\mathcal{K}$).
	
	\textbf{Case~1:} 
	$C$ separates $a$ from $b$ and $c$. 
	
	\textbf{Case~2:} 
	$C$ separates $b$ from $a$ and $c$. 
	
	\textbf{Case~3:} 
	$C$ separates $c$ from $a$ and $b$. 
	
	\textbf{Case~4:} 
	$C$ does not separate $a, b, c$.
	
	\begin{figure}[htb]
		\centering
		\hbox{}
		\hfill
		\begin{subfigure}[b]{.35\textwidth}
			\centering
			\includegraphics[page=2]{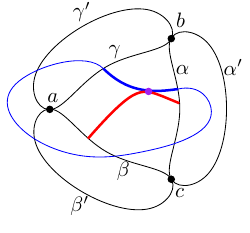}
			\caption{}
			\label{fig:K3_arc_type_proof_1}
		\end{subfigure}
		\hfill
		\begin{subfigure}[b]{.35\textwidth}
			\centering
			\includegraphics[page=3]{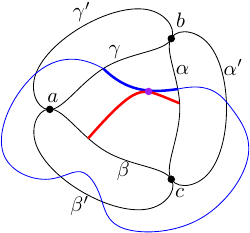}
			\caption{}
			\label{fig:K3_arc_type_proof_2}
		\end{subfigure}
		\hfill
		\hbox{}
		
		\hbox{}
		\hfill
		\begin{subfigure}[b]{.35\textwidth}
			\centering
			\includegraphics[page=4]{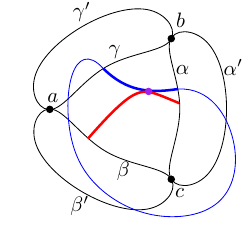}
			\caption{}
			\label{fig:K3_arc_type_proof_3}
		\end{subfigure}
		\hfill
		\begin{subfigure}[b]{.35\textwidth}
			\centering
			\includegraphics[page=5]{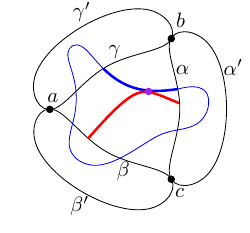}
			\caption{}
			\label{fig:K3_arc_type_proof_4}
		\end{subfigure}
		\hfill
		\hbox{}
		
		\caption{
			\subref{fig:K3_arc_type_proof_1}--\subref{fig:K3_arc_type_proof_4} illustrate Cases~1--4 from the proof of Claim~\ref{claim:same_arc_type}.
			The pseudocircles $C$ and $C'$ are highlighted blue and red, respectively. The pc-arcs $A$ and $A'$ are emphasized.
		}
		\label{fig:K3_arc_type_proof}
	\end{figure}

	\noindent
	In the next paragraph we show that in neither case is it possible to extend the arc $A'$ 
	to a pseudocircle $C'$ intersecting the three pseudocircles of $\mathcal{K}$. This is a contradiction. 
	
	Extend $A'$ starting from its endpoint on $\alpha$. The only way to reach $\gamma$ or~$\gamma'$, avoiding an invalid, additional intersection with $C$, is via the pseudocircle  \mbox{$\beta \cup \beta'$}. But the other endpoint of $A'$ already lies on $\beta$, so either the pseudocircle extending $A'$ has at least three intersections with $\beta \cup \beta'$ or it misses 
	$\gamma\cup\gamma'$. Both are prohibited in an intersecting arrangement extending~$\mathcal{K}$.
 
	This completes the proof of Claim~\ref{claim:same_arc_type}.\qedclaim\medskip

	Next we transform $\AA$ into another intersecting arrangement~$\AA'$ by redrawing the pc-arcs 
	within $\triangle$ and~$\triangle'$ such that the pairwise intersections and touchings are preserved 
    and all crossings and touchings of each arc type are concentrated in a narrow region as depicted in Figure~\ref{fig:run_example:transformed_triangles}.
	First we apply an appropriate homeomorphism on the drawing so that $\triangle$ becomes a proper triangle ($\triangle'$ will be treated in an analogous manner); see again Figure~\ref{fig:run_example:original_triangles} and Figure~\ref{fig:run_example:arc_type_triangles}. 
	For the arc type $\alpha\beta$  
	we place a small rectangular region $R_{\alpha\beta}$ within $\triangle$ that lies close to the vertex~$c$.
	We now redraw all pc-arcs of type $\alpha\beta$ so that 
	\begin{itemize}
	    \item all crossings and touchings between pc-arcs of type $\alpha\beta$ lie inside~$R_{\alpha\beta}$,
	    \item every pc-arc of type $\alpha\beta$ intersects $R_{\alpha\beta}$ on opposite sites,
	    and
	    \item
	    for every pc-arc of type $\alpha\beta$,
	    the removal of~$R_{\alpha\beta}$ leaves two straight line segments
	    which connect $R_{\alpha\beta}$ to $\alpha$ and $\beta$ (i.e., the boundary segments of~$\triangle$).
	\end{itemize}
	We proceed analogously for the arc types $\alpha\gamma$ and $\beta\gamma$.
	By Claim~\ref{claim:same_arc_type} touchings and double crossings only occur between pc-arcs of the same type and therefore lie in the rectangular regions.
	Since the rectangular regions are placed close enough to the vertices $a,b,c$ of the triangle~$\triangle$,
	no additional intersections or touching points are introduced and we obtain an arrangement~$\AA'$ of pseudocircles with the same intersections and touchings as~$\AA$.
	The combinatorics of the resulting arrangement $\AA'$ may however differ from $\AA$ since the transformation typically changes the intersection orders of the pseudocircles. We conclude:
	
	\begin{observation}
	    The transformation preserves the incidence relation between any pair of pc-arcs, that is, 
	    two pc-arcs in $\AA$ are disjoint/cross in one point/cross in two points/touch if and only if the two corresponding pc-arcs in $\AA'$ are disjoint/cross in one point/cross in two points/touch.
	\end{observation}
	
	This implies that $\AA'$ is indeed again an arrangement of $n(\AA')=n(\AA)$ pairwise intersecting pseudocircles with identical touching graph $\TG(\AA')=\TG(\AA)$. In particular, the number of touchings is preserved.

\begin{figure}[htb]
    \centering
    \includegraphics[page=1]{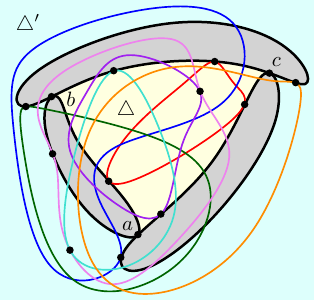}
    \caption{Original arrangement~$\AA$}
    \label{fig:run_example:original_arrangement}
\end{figure}

\begin{figure}[htb]
    \centering
    \begin{subfigure}[b]{.49\textwidth}
        \centering
        \includegraphics[page=2]{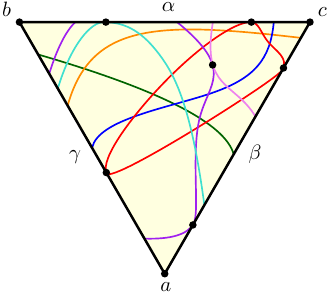}
        \caption{}
        \label{fig:run_example:inside_original}
    \end{subfigure}
    \hfill
    \begin{subfigure}[b]{.49\textwidth}
        \centering
        \includegraphics[page=5]{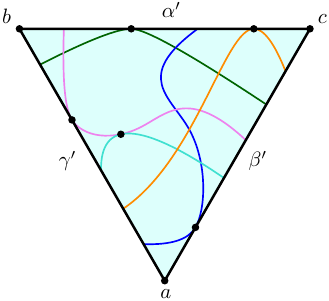}
        \caption{}
        \label{fig:run_example:outside_original}
    \end{subfigure}
    \caption{
        Area~$\triangle$~\subref{fig:run_example:inside_original} and area~$\triangle'$~\subref{fig:run_example:outside_original}, colors as in~$\AA$.
    }
    \label{fig:run_example:original_triangles}
\end{figure}

\begin{figure}[htb]
    \centering
    \begin{subfigure}[b]{.49\textwidth}
        \centering
        \includegraphics[page=3]{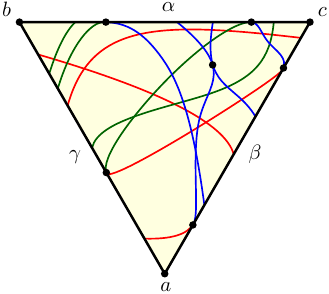}
        \caption{}
        \label{fig:run_example:inside_colored}
    \end{subfigure}
    \hfill
    \begin{subfigure}[b]{.49\textwidth}
        \centering
        \includegraphics[page=6]{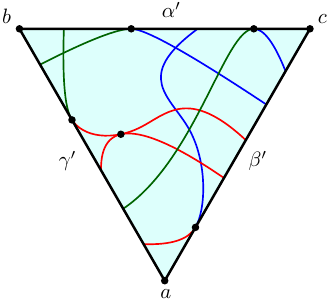}
        \caption{}
        \label{fig:run_example:outside_colored}
    \end{subfigure}
    \caption{
        Area~$\triangle$~\subref{fig:run_example:inside_transformed} and area~$\triangle'$ \subref{fig:run_example:outside_transformed}, colors according to arc-type.
    }
    \label{fig:run_example:arc_type_triangles}
\end{figure}

\begin{figure}[htb]
    \centering
    \begin{subfigure}[b]{.49\textwidth}
        \centering
        \includegraphics[page=4]{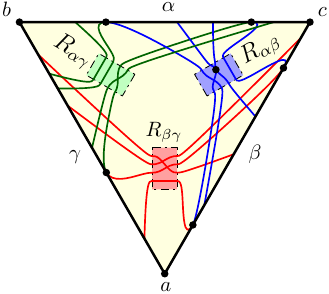}
        \caption{}
        \label{fig:run_example:inside_transformed}
    \end{subfigure}
    \hfill
    \begin{subfigure}[b]{.49\textwidth}
        \centering
        \includegraphics[page=7]{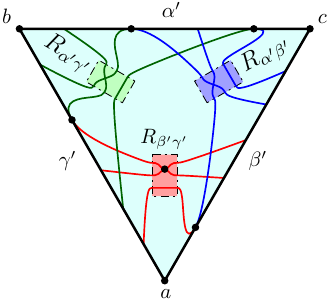}
        \caption{}
        \label{fig:run_example:outside_transformed}
    \end{subfigure}
    \caption{
        Area~$\triangle$~\subref{fig:run_example:inside_transformed} and area~$\triangle'$ \subref{fig:run_example:outside_transformed}. Concentrate all crossings and touchings of one arc type in a narrow region. The~narrow regions are indicated by dashed rectangles.
    }
    \label{fig:run_example:transformed_triangles}
\end{figure}

\begin{figure}[htb]
    \centering
    \includegraphics[page=8]{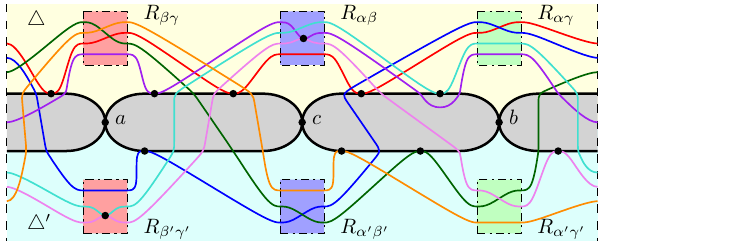}
    \caption{A cylindrical drawing of~$\AA'\setminus\KK$}
    \label{fig:run_example:almost_cylindrical_drawing}
    \vspace{0.2cm}
    \includegraphics[page=9]{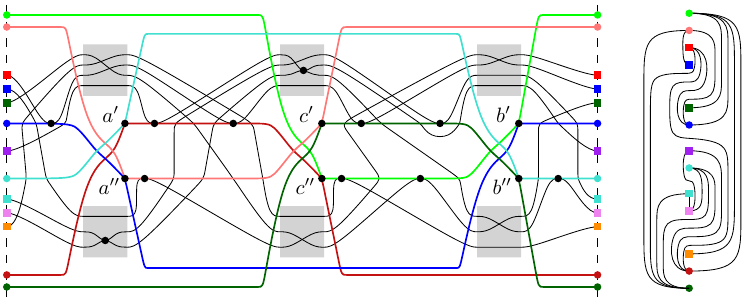}
    \caption{Replace each of the three pseudocircles of $\KK$ by two new pseudocircles so that the entire arrangement is now cylindrical. The pseudocircle from Figure~\ref{fig:run_example:almost_cylindrical_drawing} that contains the points~$a$ and~$c$ (resp. points~$c$ and~$b$ / points~$b$ and~$a$) is replaced by a new dark red and a new bright red (resp.\ dark green and bright green / dark blue and turqoise) pseudocircle. The right part shows the corresponding touching graph.}
    \label{fig:run_example:extended_to_cylindrical_drawing}
\end{figure}

    

	\begin{claim}
	\label{claim:cylindrical}
	The arrangement induced by $\AA' \setminus \KK$ is cylindrical.
	\end{claim}
	
    \noindent\textit{Proof of Claim~\ref{claim:cylindrical}.} For each pseudocircle $C\in\AA'\setminus\KK$, the intersection $$C \cap (\triangle \cup \triangle') = (C \cap \triangle) \cup (C \cap \triangle')$$
	consists of three pc-arcs,
	and each of these three pc-arcs is of a different type.
	The first arc is of type $\alpha\beta$ or $\alpha'\beta'$ (depending on whether it is inside $\triangle$ or $\triangle'$), 
	the second is of type $\beta\gamma$ or~$\beta'\gamma'$, and 
	the third is of type $\alpha\gamma$ or~$\alpha'\gamma'$.

	Now we redraw $\AA'$ on a cylinder 
	as illustrated in  Figure~\ref{fig:run_example:almost_cylindrical_drawing}.
	Since all crossings and touchings of the arc type
	are within a small region,
	all pseudocircles from $\AA' \setminus \KK$ wrap around the cylinder,
	and hence the arrangement induced by $\AA' \setminus \KK$ is cylindrical.	
 
	This completes the proof of Claim~\ref{claim:cylindrical}.\qedclaim\medskip
	
	Next we replace the three pseudocircles of $\KK$ 
	by six pseudocircles 
	as illustrated in the left part of Figure~\ref{fig:run_example:extended_to_cylindrical_drawing},
	so that the resulting arrangement $\AA''$ 
	is cylindrical. Each of the three touching points $a,b,c$ in~$\KK$ 
	is replaced by two new touching points and altogether we obtain  touchings~$a',a'',b',b'',c',c''$.
	Hence, when transforming $\AA$ into $\AA''$, the number of pseudocircles is increased by~3 and the number of touchings is also increased by~3. 
	
	Agarwal~et~al.~\mbox{\cite{AgarwalNPPSS2004}} proved the $p_2 \leq 2n-3$ upper bound on the number of touchings in cylindrical arrangements of $n$ pairwise intersecting pseudocircles 
	by bounding the number of touchings in an arrangement of pairwise intersecting pseudoparabolas.
	They show that their touching graph is planar and bipartite~\mbox{\cite[Theorem 2.4]{AgarwalNPPSS2004}}, hence, it has at most $2n-4$ edges. The difference between $2n-4$ and $2n-3$ comes from the fact that the upper or the lower face in the pseudoparabola drawing of a pseudocircle arrangement $\AA$ can be a digon of $\AA$.
    
	The drawing of~$\AA''$ in Figure~\ref{fig:run_example:extended_to_cylindrical_drawing} can be seen as an intersecting
    arrangement of pseudoparabolas. The complexity of the upper and of the lower face is three, hence, the arrangement has at most $2n(\AA'')-4$ touchings.
    
    We review the ideas of the proof of Agarwal~et~al.~\mbox{\cite{AgarwalNPPSS2004}} to verify the following claim.
	
	\begin{claim}
	\label{claim:planar_bipartite}
	$\TG(\AA'')$ is planar, bipartite, and has at most $2n(\AA'')-5$ edges.
	\end{claim}
	
    \noindent\textit{Proof of Claim~\ref{claim:planar_bipartite}.} Label the pseudoparabolas $P_1,\ldots,P_n$ such that the 
    starting segments are ordered from top to bottom.
    In the touching graph $\TG(\AA'')$, 
    we label the corresponding vertices as~$1,\ldots,n$.
    
    \textbf{Bipartiteness:}
    The bipartition comes from the fact that
    the digons incident to a fixed pseudoparabola $P_j$
    are either all from below or all from above.
    Suppose that 
    a pseudoparabola~$P_j$ has a touching from above with $P_i$ and from below with~$P_k$. It follows that $P_i$ is above $P_j$ everywhere and $P_k$ is below $P_j$ everywhere. Hence, $P_i$ and $P_k$ are separated by $P_j$ and cannot intersect -- this contradicts the assumption that the pseudocircles are pairwise intersecting.
	
	We now further observe that
	the uppermost pseudoparabola~$P_1$ and the lowermost pseudoparabola~$P_n$ belong to distinct parts of the bipartition, because
	$P_1$ has all touchings below (i.e. with parabolas of greater index);
	$P_n$ has all touchings above (i.e. with parabolas of smaller index).
	Hence, the touching graph remains bipartite after adding the edge~$\{1,n\}$.

	\textbf{Planarity:}
	For the planarity of $\TG(\AA'')$,
	Agarwal et al.\ \cite{AgarwalNPPSS2004}
	create a particular drawing:
	The vertices are drawn on a vertical line and
	each edge $e=\{u,v\}$ is drawn as an~$y$-monotone curve
	according to the following \emph{drawing rule}:
	For each $w$ with $u<w<v$,
	we route $e$ to the left of~$w$ 
	if the pseudoparabola $P_w$ intersects $P_u$ before $P_v$, otherwise we draw the edge $e$ right of~$w$.
	It is then shown that in the so-obtained drawing~$\DD$, 
	each pair of independent edges 
	has an even number of intersections. The right part of Figure~\ref{fig:run_example:extended_to_cylindrical_drawing} shows such a drawing of the corresponding touching graph. The Hanani--Tutte theorem (cf.\ Section~3 in \cite{Schaefer2013}) 
	implies that~$\TG(\AA'')$ is planar.
	
	Notice that $\{1, n\}$ is not an edge in $\TG(\AA'')$, since by construction, the lowermost and uppermost pseudocircles do not touch. We further observe that,
	since all edges in~$\DD$ are drawn as $y$-monotone curves,
	the entire drawing lies in 
	a box which is bounded from above by vertex~$1$ 
	and from below by vertex~$n$.
	Hence, we can draw an additional edge from $1$ to $n$ 
	which is routed entirely outside of the box and does not intersect any other edge.
	Again, by the Hanani--Tutte theorem, we have planarity.
	Since any planar bipartite graph on $n$ vertices 
	has at most $2n-4$ edges,
	we conclude that $\TG(\AA'')$ has at most $2n-5$ edges.

 This completes the proof of Claim~\ref{claim:planar_bipartite}.\qedclaim\medskip
	
	We are now ready to finalize the proof of Theorem~\ref{thm:digons}.
	  From Claim~\ref{claim:planar_bipartite} and $n(\AA'') = n+3$ we get~$ p_2(\AA'') \le 2(n+3)-5$. Since $p_2(\AA'') = p_2(\AA)+3$ this 
        implies~$p_2(\AA) \le 2n-2$, which is the desired bound.
\end{proof}

	\section{Proof of Proposition~\ref{prop:trianglefree_tight}}
	\label{sec:prop:trianglefree_tight}

The proof of Proposition \ref{prop:trianglefree_tight} is based on the \textit{blossom operation}, which allows to dissolve triangles in the touching graph. We will apply the blossom operation to arrangements whose touching graphs are wheel graphs to obtain arrangements with the  desired properties.

    Let $\AA$ be an arrangement of pairwise intersecting pseudocircles, let~$v$ be a pseudocircle in~$\AA$, and let $w_1, \dots, w_d$ be the pseudocircles in $\AA$ which form touchings with~$v$ in this particular circular order along~$v$. Since $\AA$ is intersecting, all the touchings are on the same side of~$v$. As illustrated in Figure~\ref{fig:blossom_operation_circles}, the blossom operation relaxes the touchings between~$v$ and~$w_1,\ldots,w_d$ to digons and inserts~$d$ new pseudocircles~$v_1', \dots, v_d'$ inside and very close to~$v$ so that 
    \begin{itemize}
        \item 
        $v_1',\ldots,v_d'$ form a cylindrical arrangement,
        \item 
        $v$ touches~$v_1', \dots, v_d'$, and
        \item 
        $w_i$ touches $v_{i-1}'$ and $v_{i}'$	(indices modulo $d$).
    \end{itemize}
	
    Since the new pseudocircles $v_1',\ldots,v_d'$ are added in an~\mbox{$\varepsilon$-small} area close to~$v$,  it is ensured that each $v_i'$ intersects all other pseudocircles. Hence, the obtained arrangement is again an arrangement of pairwise intersecting pseudocircles.  
	
    \begin{figure}[htb]	
        \centering
        \includegraphics{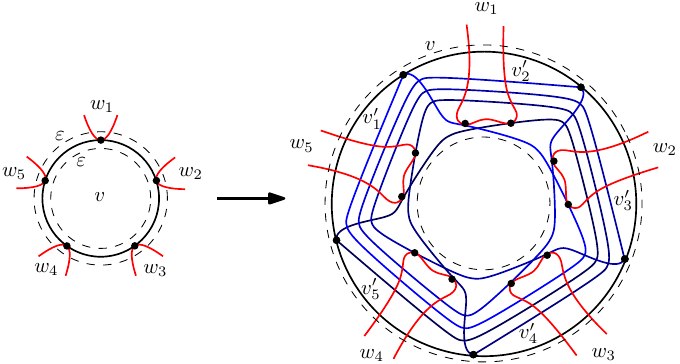}
        \caption{An illustration of the blossom operation applied on the pseudocircle $v$ of an arrangement.
        }
        \label{fig:blossom_operation_circles}
    \end{figure}
	
	Figure~\ref{fig:blossom_operation_graph} shows the effect of the blossom operation on the touching graph. Note that in these graph drawings the circular orders of the edges incident to a vertex coincide with the orders in which the touchings appear on the corresponding pseudocircle. 
	
	\begin{figure}[htb]    	
    	\centering
    	\includegraphics{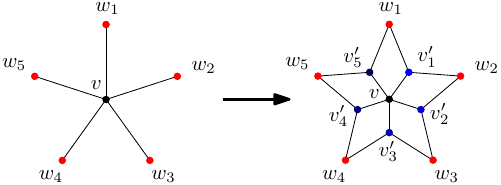}
    	\caption{Blossom operation applied on $v$: Modification of the touching graph.}
    	\label{fig:blossom_operation_graph}
	\end{figure}
	
	The blossom operation increases the number $n(\AA)$ of pseudocircles of arrangement~$\AA$ by~$d$ and the number~$p_2(\AA)$ of touchings by $2d$. Hence, when applied to an arrangement $\AA$ with exactly~$p_2(\AA) = 2n(\AA) - 2$ touchings, the blossom operation again yields an arrangement $\AA'$ with $p_2(\AA') = 2n(\AA') - 2$ touchings.
	
	The blossom operation can be used to eliminate triangles in the touching graph. Assume pseudocircles $w_i$ and $w_j$ touch, 
 hence $v, w_i, w_j$ form a triangle in the touching graph. Then the blossom operation on $v$ destroys this triangle without creating a new one if and only if, along the pseudocircle~$v$, the two touchings with $w_i$ and $w_j$ are not consecutive.  In Figure~\ref{fig:blossom_operation_graph} a triangle $\{v, w_1, w_2\}$ would result in the new triangle $\{v_1', w_1, w_2 \}$, while a triangle~$\{v, w_1, w_3\}$ would be eliminated without replacement.
	
	\medskip
	
	Using the blossom operation, 
	we are now able to prove Proposition \ref{prop:trianglefree_tight}.

 \propTriangleFreeTight*
\begin{proof}
    Let $n' \ge 11$ be an integer with $n' \equiv 3 \pmod{4}$.
    Then $n = \frac{n'+1}{2}$ is an even integer with~$n \ge 6$.
    As illustrated in  Figure~\ref{fig:wheelexample_A} and Figure~\ref{fig:wheelexample_B},
    we can construct an
    arrangement~$\AA$ of~$n$ pseudocircles with 
    \mbox{$p_2=2n-2$} touchings such that the touching graph~$T(\AA)$ is the wheel graph~$W_n$.
    
    \begin{figure}[htb]
        \centering
        
        \hbox{}
        \hfill
        \begin{subfigure}[b]{.48\textwidth}
            \centering
            \includegraphics[page=2,width=0.8\textwidth]{figs/wheelexample_pseudocircles.pdf}
            \caption{}
            \label{fig:wheelexample_A}
        \end{subfigure}
        \hfill
        \begin{subfigure}[b]{.48\textwidth}
            \centering
            \includegraphics[page=1,width=0.8\textwidth]{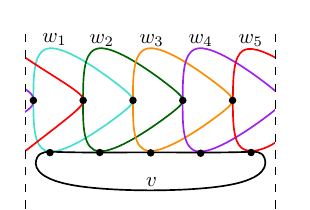}
            \vspace{0.6cm} 
            \caption{}
            \label{fig:wheelexample_B}
        \end{subfigure}
        \hfill
        
        \hbox{}
        \hfill
        \begin{subfigure}[b]{.48\textwidth}
            \centering
            \includegraphics{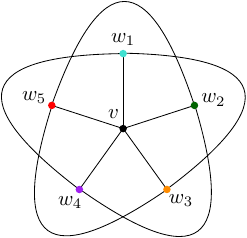}
            \caption{}
            \label{fig:wheelexample_C}
        \end{subfigure}
        \hfill
        \begin{subfigure}[b]{.48\textwidth}
            \centering
            \includegraphics{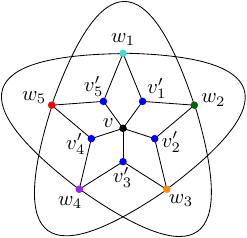}
            \caption{}
            \label{fig:wheelexample_D}
        \end{subfigure}
        \hfill
        
        \caption{
            \subref{fig:wheelexample_A}~An arrangement $\AA$ of $6$ pseudocircles, 
            \subref{fig:wheelexample_B}~its cylindrical representation,
            \subref{fig:wheelexample_C}~its touching graph $T(\AA)$, and
            \subref{fig:wheelexample_D}~the touching graph $T(\AA')$ after applying the blossom operation to~$v$.
        }
        \label{fig:wheelexample}
    \end{figure}
    
    In this construction
    the \emph{central} pseudocircle $v$ has a touching with each of the pseudocircles~$w_1,\ldots,w_{n-1}$
    and each $w_i$ touches $v$, $w_{i+n/2}$, and $w_{i-n/2}$ (indices modulo~$n-1$);
    see Figure~\ref{fig:wheelexample_C}. 
    
    All triangles in~$\TG(\AA)$ contain the central vertex~$v$
    and
    for each such triangle $\{v, w_i, w_j\}$, the touchings of the pseudocircles $w_i$ and~$w_j$ with the pseudocircle~$v$ are not consecutive on~$v$. 
    Therefore, applying the blossom operation to~$v$ eliminates all triangles and the resulting arrangement $\AA'$ of $n'=2n-1$ pairwise intersecting pseudocircles has $p_2(\AA')=2n'-2$ touchings and a triangle-free touching graph $\TG(\AA')$; see Figure~\ref{fig:wheelexample_D}.
    This completes the argument for~$n' \ge 11$ with $n \equiv 3 \pmod{4}$.
    
    To give a construction for $n''=14$ and for all integers $n'' \ge 17$,
    note that the blossom operation can be applied to pseudocircles with exactly three touchings.
    The constructed examples with $n \equiv 3 \pmod{4}$ have pseudocircles with three touchings and the blossom operation applied to such a pseudocircle preserves the property.
    
    Since $n''=14$ and every integer $n'' \ge 17$ can be written as $n' + 3k$ with \mbox{$n' \in \{11,15,19\}$} and~$k \in \mathbb{N} \cup \{0\}$, we obtain arrangements $\AA''$ of $n''$ pseudocircles with $p_2(\AA'')=2n''-2$ touchings. This completes the proof of Proposition~\ref{prop:trianglefree_tight}.
\end{proof}

    \section{Digon-free arrangements with few triangles}
    \label{sec:thm:triangles}
	
    \subsection{Proof of Theorem~\ref{thm:triangles} and Theorem~\ref{thm:extension_to_counterexample}}
    \label{subsec:proof_of_triangle_thms}

    The proofs of Theorem~\ref{thm:triangles} and Theorem~\ref{thm:extension_to_counterexample} are both based on replacing pseudocircles by canonical bundles of $4$ pseudocircles, as shown in Figure~\ref{fig:bundle_operation}. Like in the blossom operation, the new pseudocircles are placed within an $\varepsilon$-small area around the replaced pseudocircle so that the intersecting property of the pseudocircles is being preserved. We call such an operation a \emph{bundle replacement} and aim for performing them on pseudocircles in order to destroy some of their incident triangles.

    \begin{figure}[htbh]
        \centering
        \includegraphics{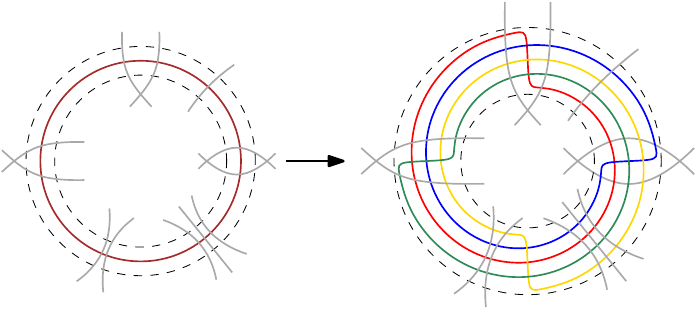}
        \caption{Bundle replacement: a pseudocircle is replaced by a bundle of $4$ pseudocircles.}
        \label{fig:bundle_operation}
    \end{figure}

    The following fact is a direct consequence of a sweepability statement by Snoeyink and Hershberger~\cite{SnoeyinkHershberger1991}.

    \begin{proposition}[{\cite[Lemma 4.1]{SnoeyinkHershberger1991}}]\label{prop:two_incident_triangles}
        In a digon-free arrangement of pairwise intersecting pseudocircles, each pseudocircle is incident to at least two triangles to the inside and two triangles to the outside.
    \end{proposition}

    \thmExtensionToCounterexample*

    \begin{proof}
        Let~$\AA$ be a digon-free arrangement of intersecting pseudocircles.
        Select a pseudocircle~$p$ in~$\AA$. Figure~\ref{fig:bundle} illustrates a situation which might have been obtained via a bundle replacement on~$p$. Note that a bundle replacement leads to a new arrangement~$\AA'$ which is again digon-free and contains $\AA$ as a subarrangement.

        We can think of the new bundle as being composed of four sections that are delimited by the four twists starting and ending in the purple crossings; colors refer to Figure~\ref{fig:bundle}. For a precise description we define a \emph{twist} in a bundle as a sequence of consecutive crossings which make the outermost pseudocircle of the bundle the innermost. In all of our figures we keep the crossings of a twist close together.
        
        Proposition~\ref{prop:two_incident_triangles} guarantees that~$p$ is incident to at least~$4$ triangles. Observe that the four twists in the bundle can always be distributed in such a way that each of these~$4$ triangles becomes incident to one of the twists, hence, the triangles are turned into quadrangles (green cells). On the other hand, in each of the~$8$ red areas, independent of the number of crossing pseudocircles (gray), exactly one new triangle is created.
        
        In total, a careful bundle replacement on~$p$ leads to  a digon-free arrangement~$\AA'$ having at most~$p_3(\AA')\leq p_3(\AA)+4$ triangles and which contains~$\AA$ as a subarrangement. This procedure can be iterated. Each iteration increases the number of pseudocircles by~$3$ and the number of triangles by at most~$4$. 

        Let $t$ be the number of triangles of the initial arrangement $\AA$.
        If $m \geq t/3\varepsilon$, then the arrangement $\hat\AA$ obtained 
        through a sequence of $m$ bundle 
        replacements has the claimed property: \[p_3(\hat\AA) \leq t  + 4m \leq (4/3 + \varepsilon)3m < (4/3 + \varepsilon)\hat n .\]
    \end{proof}

    \begin{figure}[htb]
        \centering
        \includegraphics{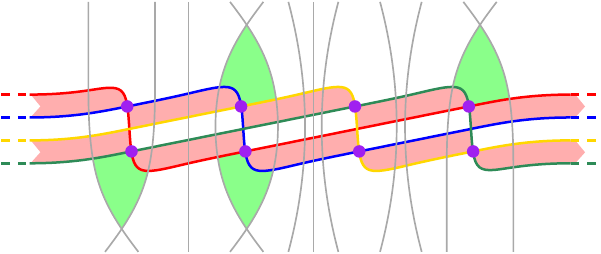}
        \caption{Situation obtained by a bundle replacement.}
        \label{fig:bundle}
    \end{figure}
 
    \begin{figure}[htb]
        \centering
        \includegraphics[width=.66\textwidth]{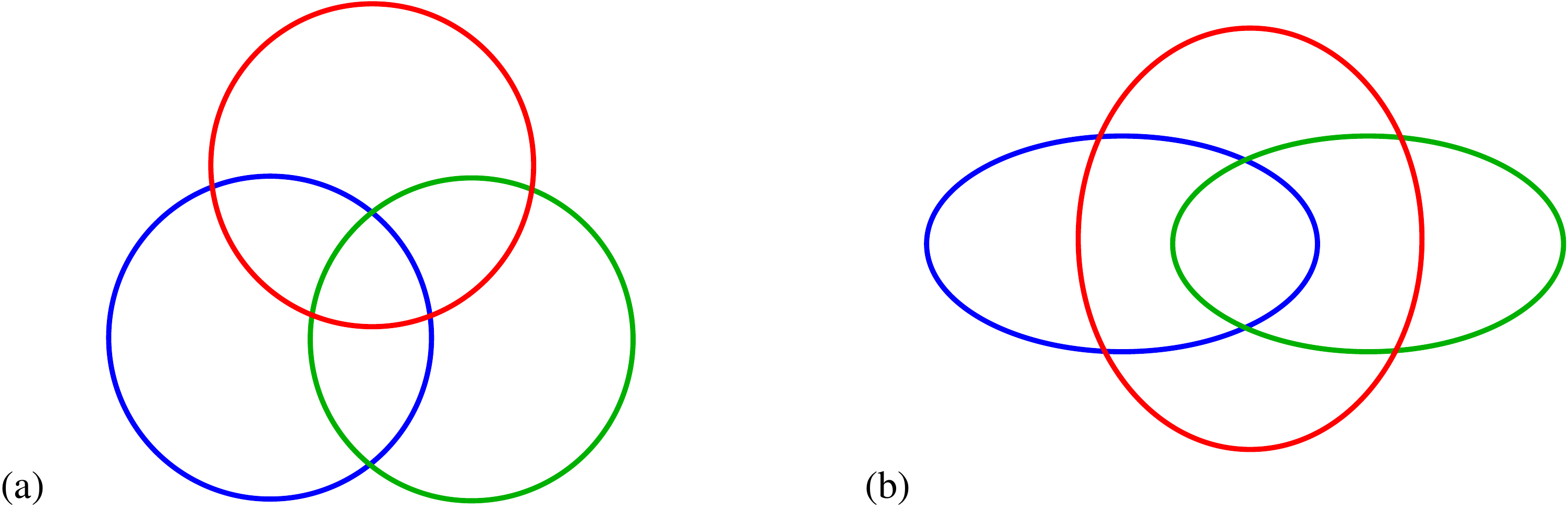}
        \caption{The two intersecting simple arrangements of three pseudocircles. (a) Krupp, (b)~NonKrupp.}
        \label{fig:Krupp}
    \end{figure}

    \thmTriangles*

    \begin{proof} 
    Figure~\ref{fig:Krupp} shows the two intersecting simple arrangements of three pseudocircles. The Krupp is digon free and has 
    8 triangle cells. Note that the Krupp can be obtained from a single isolated circle in a bundle replacement step with a bundle of size three. Replacing any of the circles of the Krupp
    with a bundle of size four we can convert all the original triangles to 4-gons while generating eight new triangles. The result is the arrangement~$\AAsixA$ from Figure~\ref{fig:AAsixA} with~6~pseudocircles and eight triangles. Starting from~$\AAsixA$,
    we can iterate the bundle replacement with bundles of size four; this yields a family of arrangements~$\AA_{3k}$ with~$n=3k$ pseudocircles and~$4k = \frac{4}{3}n$ triangles.  

  \begin{figure}[htb]
        \centering
        \includegraphics[page=2,width=0.35\textwidth]{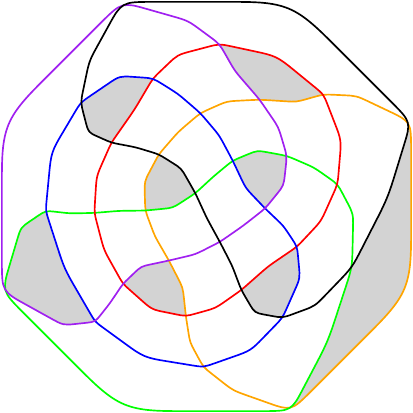}
        \caption{The non-circularizable digon-free intersecting arrangement $\AAsixA$. 
        \label{fig:AAsixA}}
    \end{figure}

    For values of $n$ which are not divisible by 3, we can take the arrangement~$\AA_{3k}$ with~$k= \lfloor\frac{n}{3}\rfloor$ and apply a bundle replacement step with a bundle of size two (for~$n=3k+1$) or three (for~$n=3k+2$). In the first case, we can eliminate two old triangles at the cost of creating four new ones; hence, the new arrangement~$\AA_{3k+1}$ has~$n=3k+1$ pseudocircles and~$4k +2 = \lceil\frac{4}{3}n\rceil$ triangles. In the second case,  
    we can eliminate three old triangles at the cost of creating six new ones and obtain a new arrangement $\AA_{3k+2}$ with $n=3k+2$ pseudocircles and~$4k +3 = \lceil\frac{4}{3}n\rceil$ triangles.
    \end{proof}

\subsection{Proof of Proposition~\ref{prop:nosubN6family}}
\label{ssec:no-N6}

In this subsection we construct intersecting digon-free arrangements with few triangles (less than $2n-4$) and no $\AAsixA$ subarrangement. The key to the construction is again the replacement of pseudocircles of a base arrangement by bundles. In the following proof we use bundles of size 3. In the discussion we will also mention size 4 and larger sizes.

\propNoSubNSixFamily*

\begin{proof}
    For any $N\geq 2$, let $\mathcal{A}_N$ be the  arrangement of pseudocircles~$C_1, \cdots, C_N$ 
    such that the cyclic order of intersection of $C_i$ with the other pseudocircles is:
    $$
    1,2,\ldots,i-1,i+1,\ldots,N,
    N,N-1,\ldots,i+1,i-1,\ldots,1.
    $$
    The arrangement can be represented with axis-parallel rectangles so that for all $i < j$ the right side of $C_i$ cuts vertically through $C_j$; see Figure~\ref{fig:n6free_construction:without_bundles} for $\mathcal{A}_4$. Note that every triple of three pseudocircles of $\mathcal{A}_N$ induces a NonKrupp. The faces of $\mathcal{A}_N$ are:
    3 digons,~$2(N-2)$ triangles,~$N(N-3)+2$ four-gons, and a single~$(2N-2)$-gon.

  On the basis of $\mathcal{A}_N$ we construct an arrangement $\mathcal{B}_N$ by bundle replacements. 

\begin{figure}[htb]
    \centering
    
    \begin{subfigure}[b]{.34\textwidth}
        \centering
        \includegraphics[page=2,scale=0.8]{figs/n6free_construction.pdf}\vspace{0.6cm}
        \caption{The arrangement $\mathcal{A}_4$ }
        \label{fig:n6free_construction:without_bundles}
    \end{subfigure}
    \hspace{0.3cm}
    \begin{subfigure}[b]{.59\textwidth}
        \centering
        \includegraphics[page=1,scale=0.65]{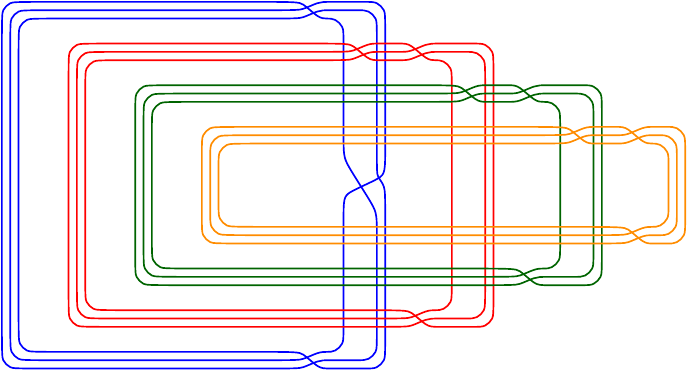}
        \caption{The arrangement $\mathcal{B}_4$.}
        \label{fig:n6free_construction:with_bundles}
    \end{subfigure}
    
    \caption{Replacing the pseudocircles with bundles so that the twists destroy all digons.}
    \label{fig:n6free_construction}
\end{figure}
   
   From $\mathcal{A}_N$ we obtain the arrangement $\mathcal{B}_N$ by successively replacing $C_1, \cdots, C_N$ by bundles~$B_1, \cdots, B_N$ of size~$3$
    where we place the twists according to the following rules; Figure~\ref{fig:n6free_construction:with_bundles} shows the arrangement~$\mathcal{B}_4$ obtained by using the base arrangement $\mathcal{A}_4$ from Figure~\ref{fig:n6free_construction:without_bundles}.
    \begin{itemize}
        \item The bundle replacing $C_1$, has two twists outside of $C_2$ and one twist inside of $C_N$.
        \item The bundle $B_i$ replacing $C_i$, $1<i<N$, has
        two twists which increase the degree of each of the two triangles formed by $B_{i-1}$, $C_i$ and $C_{i+1}$. The third twist is within $B_{i-1}$
        such that it increases the degree of a triangle corresponding to a twist of $B_{i-1}$.
        \item Finally, when replacing $C_N$, place two consecutive twist which increase the degree of the digon formed by $B_{N-1}$ and $C_N$ and
        one twist in a triangle corresponding to a twist of~$B_{N-1}$.
    \end{itemize}
    The arrangement~$\mathcal{B}_N$ consists of $n=3N$ pseudocircles. For counting the triangles, each bundle~$B_i$ has $6$ internal triangles, but for the bundles $B_1, \cdots, B_{N-1}$ one of them is dissolved by a twist of the next bundle $B_{i+1}$.
    Additionally, the digon formed by $C_1$ and $C_N$ has been made a triangle. The total number of triangles of  $\mathcal{B}_N$ is
    \[ 
    6N - (N-1) + 1 = 5N+2 = \frac{5}{3}n + 2.
    \] 
    The arrangement $\mathcal{B}_N$ is digon-free. It remains to show that there is no subarrangement of~$\mathcal{B}_N$ which is isomorphic to~$\AAsixA$. 
    
    Since every triple of pseudocircles of~$\mathcal{A}_N$ induces a NonKrupp, the same is true for triples of pseudocircles taken from distinct bundles of~$\mathcal{B}_N$.
    A~$\AAsixA$ arrangement has exactly $4$ triples that form a NonKrupp and each of the six pseudocircles is member of exacly two of them.
    
    Let~$\mathcal{B}'$ be a subarrangement of $\mathcal{B}_N$ consisting of~$6$ pseudocircles. If~$\mathcal{B}'$ contains pseudocircles of at least four different bundles then it has strictly more than 4 NonKrupps, hence, it cannot be isomorphic to~$\AAsixA$. Now assume that~$\mathcal{B}'$ consists of~$k_1, k_2, k_3\geq 0$ pseudocircles from three pairwise different bundles,~$k_1 + k_2 + k_3 = 6$.

    \begin{claim}
        If $k_i>0$ for~$i=1,2,3$, then~$\mathcal{B}'$ is not isomorphic to~$\AAsixA$.
    \end{claim}

    \noindent\textit{Proof of Claim~1.}
    Each triple of pseudocircles of pairwise different bundles forms a NonKrupp, which implies that~$\mathcal{B}'$ contains at least $k_1\cdot k_2 \cdot k_3$ NonKrupp subarrangements. As~$k_1+k_2+k_3=6$, this value is~$4$,~$6$, or $8$.  If $k_1\cdot k_2\cdot k_3=4$, two of the $k_i$ are equal to $1$ and the two corresponding pseudocircles participate in all 4 NonKrupp subarrangements of $\mathcal{B}'$. Hence,
    in all cases $\mathcal{B}'$ and~$\AAsixA$ are not isomorphic.
    \qedclaim\medskip
    
    It follows that if~$\mathcal{B}'$ is isomorphic to $\AAsixA$, then it must contain the three pseudocircles of each of two bundles, i.e., $k_1=k_2=3$.

    \begin{claim}
        If $\mathcal{B}'$ is a subarrangement consisting of two complete bundles of $\mathcal{B}_N$, then $\mathcal{B}'$ is not isomorphic to $\AAsixA$.
    \end{claim}

    \noindent\textit{Proof of Claim~2.}
    From the bundle structure of~$\mathcal{B}'$ we get~$6$ triangles in each of the two bundles. A~twist of one bundle placed between consecutive pseudocircles of the other bundle can destroy one triangle of the second bundle. Such a twist corresponds to one of the two crossings in the underlying arrangement of two circles. Hence the arrangement~$\mathcal{B}'$ has at least 10 triangles, while~$\AAsixA$ has only~$8$.
    \qedclaim
\smallskip

    This excludes the existence of a subarrangement~$\mathcal{B}'$ isomorphic to $\AAsixA$.
\end{proof}

We now sketch how the constant $5/3$ of Proposition~\ref{prop:nosubN6family} could be
replaced by the smaller constant $3/2$. 

Again we take the arrangement $\mathcal{A}_N$ as a basis but now we replace each pseudocircle with a bundle of size 4. This can lead to
an intersecting digon-free arrangement with 
$n=4N$ pseudocircles  and \[
    (8N + 2(N-2) +6) - 4N = 6N+2 = \frac{3}{2}n+2
\] triangles;
the count is as follows: each of the~$4N$ twists can increase the degree of a triangle or digon by one,
the initial arrangement $\mathcal{A}_N$ has 3 digons and $2(N-2)$ triangles. 
To achieve~$\AAsixA$-freeness with bundles of size 4 the twists have to be placed with some care, and the analysis that the result is indeed $\AAsixA$-free requires the analysis of a lot of cases.
In fact~$\AAsixA$ can be obtained from the arrangement with two circles by replacing one with a bundle of size 2 and the other with a bundle of size 4 if the twists are placed in a specific way.

\medskip

Every arrangement $\AA$ with the property that every triple of pseudocircles forms a NonKrupp can be 
used as the basis for a construction with bundle replacement with bundles of size 3 and/or 4 such that the constructed arrangement 
is intersecting, digon-free, $\AAsixA$-free and has few triangles. 
Using bundles of larger size makes it more challenging to avoid $\AAsixA$-subarrangements,
and we see no way of getting below $\frac{3}{2}n$ triangles with such a construction.

For Conjecture~\ref{conj:weakGB} to be true, it would be necessary that $\AAsixA$-free 
arrangements obtained by bundle replacement with few triangles are non-circularizable. 
With the help of the \texttt{polymake}~\cite{polymake:2000} extension~\texttt{r9n} developed by Julian Pfeifle, we could verify that the arrangement~$\mathcal{B}_3$ with~9 pseudocircles and 17 triangles and the arrangement~${\cal C}_7$ with 7 pseudocircles shown in Figure~\ref{fig:arrangement_c7} are both not circularizable.
We leave the following questions for future research:

\begin{figure}[htb]
    \centering
    \includegraphics{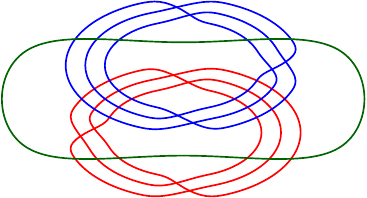}
    \caption{Non-circularizable arrangement $C_7$.}
    \label{fig:arrangement_c7}
\end{figure}

\begin{itemize}
\item What is the minimum number of triangles of 
intersecting, digon-free, $\AAsixA$-free arrangements of pseudocircles? 
\item What is the minimum number of triangles of 
intersecting, digon-free arrangements of circles?
Is it $2n-4$? (Conjecture~\ref{conj:weakGB})
\end{itemize}

    \bibliographystyle{splncs04}
    \bibliography{bibliography}
	
\end{document}